\documentclass{amsart}
\usepackage{amsfonts}
\usepackage{amsmath,amscd}
\usepackage{amsthm}
\usepackage{amssymb}
\usepackage{latexsym}
\newtheorem{lem}{Lemma}[section]

\newtheorem{defn}{Definition}[section]

\newtheorem{conj}{Conjecture}
\newtheorem{rem}{Remark}

\numberwithin{equation}{section}
\begin{document}
\newcommand{\beqa}{\begin{eqnarray}}
\newcommand{\eeqa}{\end{eqnarray}}
\newcommand{\thmref}[1]{Theorem~\ref{#1}}
\newcommand{\secref}[1]{Sect.~\ref{#1}}
\newcommand{\lemref}[1]{Lemma~\ref{#1}}
\newcommand{\propref}[1]{Proposition~\ref{#1}}
\newcommand{\corref}[1]{Corollary~\ref{#1}}
\newcommand{\remref}[1]{Remark~\ref{#1}}
\newcommand{\er}[1]{(\ref{#1})}
\newcommand{\nc}{\newcommand}
\newcommand{\rnc}{\renewcommand}

\nc{\cal}{\mathcal}

\nc{\goth}{\mathfrak}
\rnc{\bold}{\mathbf}
\renewcommand{\frak}{\mathfrak}
\renewcommand{\Bbb}{\mathbb}

\newcommand{\id}{\text{id}}
\nc{\Cal}{\mathcal}
\nc{\Xp}[1]{X^+(#1)}
\nc{\Xm}[1]{X^-(#1)}
\nc{\on}{\operatorname}
\nc{\ch}{\mbox{ch}}
\nc{\Z}{{\bold Z}}
\nc{\J}{{\mathcal J}}
\nc{\C}{{\bold C}}
\nc{\Q}{{\bold Q}}
\renewcommand{\P}{{\mathcal P}}
\nc{\N}{{\Bbb N}}
\nc\beq{\begin{equation}}
\nc\enq{\end{equation}}
\nc\lan{\langle}
\nc\ran{\rangle}
\nc\bsl{\backslash}
\nc\mto{\mapsto}
\nc\lra{\leftrightarrow}
\nc\hra{\hookrightarrow}
\nc\sm{\smallmatrix}
\nc\esm{\endsmallmatrix}
\nc\sub{\subset}
\nc\ti{\tilde}
\nc\nl{\newline}
\nc\fra{\frac}
\nc\und{\underline}
\nc\ov{\overline}
\nc\ot{\otimes}
\nc\bbq{\bar{\bq}_l}
\nc\bcc{\thickfracwithdelims[]\thickness0}
\nc\ad{\text{\rm ad}}
\nc\Ad{\text{\rm Ad}}
\nc\Hom{\text{\rm Hom}}
\nc\End{\text{\rm End}}
\nc\Ind{\text{\rm Ind}}
\nc\Res{\text{\rm Res}}
\nc\Ker{\text{\rm Ker}}
\rnc\Im{\text{Im}}
\nc\sgn{\text{\rm sgn}}
\nc\tr{\text{\rm tr}}
\nc\Tr{\text{\rm Tr}}
\nc\supp{\text{\rm supp}}
\nc\card{\text{\rm card}}
\nc\bst{{}^\bigstar\!}
\nc\he{\heartsuit}
\nc\clu{\clubsuit}
\nc\spa{\spadesuit}
\nc\di{\diamond}
\nc\cW{\cal W}
\nc\cG{\cal G}
\nc\al{\alpha}
\nc\bet{\beta}
\nc\ga{\gamma}
\nc\de{\delta}
\nc\ep{\epsilon}
\nc\io{\iota}
\nc\om{\omega}
\nc\si{\sigma}
\rnc\th{\theta}
\nc\ka{\kappa}
\nc\la{\lambda}
\nc\ze{\zeta}

\nc\vp{\varpi}
\nc\vt{\vartheta}
\nc\vr{\varrho}

\nc\Ga{\Gamma}
\nc\De{\Delta}
\nc\Om{\Omega}
\nc\Si{\Sigma}
\nc\Th{\Theta}
\nc\La{\Lambda}

\nc\boa{\bold a}
\nc\bob{\bold b}
\nc\boc{\bold c}
\nc\bod{\bold d}
\nc\boe{\bold e}
\nc\bof{\bold f}
\nc\bog{\bold g}
\nc\boh{\bold h}
\nc\boi{\bold i}
\nc\boj{\bold j}
\nc\bok{\bold k}
\nc\bol{\bold l}
\nc\bom{\bold m}
\nc\bon{\bold n}
\nc\boo{\bold o}
\nc\bop{\bold p}
\nc\boq{\bold q}
\nc\bor{\bold r}
\nc\bos{\bold s}
\nc\bou{\bold u}
\nc\bov{\bold v}
\nc\bow{\bold w}
\nc\boz{\bold z}

\nc\ba{\bold A}
\nc\bb{\bold B}
\nc\bc{\bold C}
\nc\bd{\bold D}
\nc\be{\bold E}
\nc\bg{\bold G}
\nc\bh{\bold H}
\nc\bi{\bold I}
\nc\bj{\bold J}
\nc\bk{\bold K}
\nc\bl{\bold L}
\nc\bm{\bold M}
\nc\bn{\bold N}
\nc\bo{\bold O}
\nc\bp{\bold P}
\nc\bq{\bold Q}
\nc\br{\bold R}
\nc\bs{\bold S}
\nc\bt{\bold T}
\nc\bu{\bold U}
\nc\bv{\bold V}
\nc\bw{\bold W}
\nc\bz{\bold Z}
\nc\bx{\bold X}

\nc\ca{\mathcal A}
\nc\cb{\mathcal B}
\nc\cc{\mathcal C}
\nc\cd{\mathcal D}
\nc\ce{\mathcal E}
\nc\cf{\mathcal F}
\nc\cg{\mathcal G}
\rnc\ch{\mathcal H}
\nc\ci{\mathcal I}
\nc\cj{\mathcal J}
\nc\ck{\mathcal K}
\nc\cl{\mathcal L}
\nc\cm{\mathcal M}
\nc\cn{\mathcal N}
\nc\co{\mathcal O}
\nc\cp{\mathcal P}
\nc\cq{\mathcal Q}
\nc\car{\mathcal R}
\nc\cs{\mathcal S}
\nc\ct{\mathcal T}
\nc\cu{\mathcal U}
\nc\cv{\mathcal V}
\nc\cz{\mathcal Z}
\nc\cx{\mathcal X}
\nc\cy{\mathcal Y}

\nc\e[1]{E_{#1}}
\nc\ei[1]{E_{\delta - \alpha_{#1}}}
\nc\esi[1]{E_{s \delta - \alpha_{#1}}}
\nc\eri[1]{E_{r \delta - \alpha_{#1}}}
\nc\ed[2][]{E_{#1 \delta,#2}}
\nc\ekd[1]{E_{k \delta,#1}}
\nc\emd[1]{E_{m \delta,#1}}
\nc\erd[1]{E_{r \delta,#1}}

\nc\ef[1]{F_{#1}}
\nc\efi[1]{F_{\delta - \alpha_{#1}}}
\nc\efsi[1]{F_{s \delta - \alpha_{#1}}}
\nc\efri[1]{F_{r \delta - \alpha_{#1}}}
\nc\efd[2][]{F_{#1 \delta,#2}}
\nc\efkd[1]{F_{k \delta,#1}}
\nc\efmd[1]{F_{m \delta,#1}}
\nc\efrd[1]{F_{r \delta,#1}}

\nc\fa{\frak a}
\nc\fb{\frak b}
\nc\fc{\frak c}
\nc\fd{\frak d}
\nc\fe{\frak e}
\nc\ff{\frak f}
\nc\fg{\frak g}
\nc\fh{\frak h}
\nc\fj{\frak j}
\nc\fk{\frak k}
\nc\fl{\frak l}
\nc\fm{\frak m}
\nc\fn{\frak n}
\nc\fo{\frak o}
\nc\fp{\frak p}
\nc\fq{\frak q}
\nc\fr{\frak r}
\nc\fs{\frak s}
\nc\ft{\frak t}
\nc\fu{\frak u}
\nc\fv{\frak v}
\nc\fz{\frak z}
\nc\fx{\frak x}
\nc\fy{\frak y}

\nc\fA{\frak A}
\nc\fB{\frak B}
\nc\fC{\frak C}
\nc\fD{\frak D}
\nc\fE{\frak E}
\nc\fF{\frak F}
\nc\fG{\frak G}
\nc\fH{\frak H}
\nc\fJ{\frak J}
\nc\fK{\frak K}
\nc\fL{\frak L}
\nc\fM{\frak M}
\nc\fN{\frak N}
\nc\fO{\frak O}
\nc\fP{\frak P}
\nc\fQ{\frak Q}
\nc\fR{\frak R}
\nc\fS{\frak S}
\nc\fT{\frak T}
\nc\fU{\frak U}
\nc\fV{\frak V}
\nc\fZ{\frak Z}
\nc\fX{\frak X}
\nc\fY{\frak Y}
\nc\tfi{\ti{\Phi}}
\nc\bF{\bold F}
\rnc\bol{\bold 1}

\nc\ua{\bold U_\A}

\nc\qinti[1]{[#1]_i}
\nc\q[1]{[#1]_q}
\nc\xpm[2]{E_{#2 \delta \pm \alpha_#1}}  
\nc\xmp[2]{E_{#2 \delta \mp \alpha_#1}}
\nc\xp[2]{E_{#2 \delta + \alpha_{#1}}}
\nc\xm[2]{E_{#2 \delta - \alpha_{#1}}}
\nc\hik{\ed{k}{i}}
\nc\hjl{\ed{l}{j}}
\nc\qcoeff[3]{\left[ \begin{smallmatrix} {#1}& \\ {#2}& \end{smallmatrix}
\negthickspace \right]_{#3}}
\nc\qi{q}
\nc\qj{q}

\nc\ufdm{{_\ca\bu}_{\rm fd}^{\le 0}}


\nc\isom{\cong} 

\nc{\pone}{{\Bbb C}{\Bbb P}^1}
\nc{\pa}{\partial}
\def\H{\mathcal H}
\def\L{\mathcal L}
\nc{\F}{{\mathcal F}}
\nc{\Sym}{{\goth S}}
\nc{\A}{{\mathcal A}}
\nc{\arr}{\rightarrow}
\nc{\larr}{\longrightarrow}

\nc{\ri}{\rangle}
\nc{\lef}{\langle}
\nc{\W}{{\mathcal W}}
\nc{\uqatwoatone}{{U_{q,1}}(\su)}
\nc{\uqtwo}{U_q(\goth{sl}_2)}
\nc{\dij}{\delta_{ij}}
\nc{\divei}{E_{\alpha_i}^{(n)}}
\nc{\divfi}{F_{\alpha_i}^{(n)}}
\nc{\Lzero}{\Lambda_0}
\nc{\Lone}{\Lambda_1}
\nc{\ve}{\varepsilon}
\nc{\phioneminusi}{\Phi^{(1-i,i)}}
\nc{\phioneminusistar}{\Phi^{* (1-i,i)}}
\nc{\phii}{\Phi^{(i,1-i)}}
\nc{\Li}{\Lambda_i}
\nc{\Loneminusi}{\Lambda_{1-i}}
\nc{\vtimesz}{v_\ve \otimes z^m}

\nc{\asltwo}{\widehat{\goth{sl}_2}}
\nc\ag{\widehat{\goth{g}}}  
\nc\teb{\tilde E_\boc}
\nc\tebp{\tilde E_{\boc'}}

\newcommand{\LR}{\bar{R}}
\newcommand{\eeq}{\end{equation}}
\newcommand{\ben}{\begin{eqnarray}}
\newcommand{\een}{\end{eqnarray}}

\title[Higher order relations for ADE-type $q-$Onsager algebras]{Higher order relations for ADE-type \\
generalized $q-$Onsager algebras}
\author{P. Baseilhac and T. T. Vu}
\address{Laboratoire de Math\'ematiques et Physique Th\'eorique CNRS/UMR 7350,
        F\'ed\'eration Denis Poisson, Universit\'e de Tours, Parc de Grammont, 37200 Tours, FRANCE}
\email{baseilha@lmpt.univ-tours.fr; Thi-thao.Vu@lmpt.univ-tours.fr}

\begin{abstract} Let $\{{\textsf A}_j|j=0,1,...,rank(g)\}$ be the fundamental generators of the generalized $q-$Onsager algebra $\cal O_{q}(\widehat{g})$ introduced in \cite{BB1}, where $\widehat{g}$ is a simply-laced affine Lie algebra. New relations between certain monomials of the fundamental generators - indexed by the integer $r\in\mathbb{Z}^{+}$ - are conjectured. These relations can be seen as deformed analogues of Lusztig's $r-$th higher order $q-$Serre relations associated with ${\cal U}_q({\widehat g})$, which are recovered as special cases. The relations are proven for $r\leq 5$. For $r$ generic, several supporting evidences are presented.
\end{abstract}

\maketitle

\vskip -0.7cm

{\small MSC:\ 81R50;\ 81R10;\ 81U15;\ 81T40.}

{{\small  {\it \bf Keywords}:  Generalized $q-$Onsager algebras; Coideal subalgebras; Reflection equations; $q-$Serre relations}}

\vspace{4mm}

\section{Introduction}
Introduced in \cite{BB1}, the generalized $q-$Onsager algebra ${\cal O}_q({\widehat{g}})$ associated with the affine Lie algebra $\widehat{g}$ is a higher rank generalization of the so-called $q-$Onsager algebra \cite{Ter03,B1}. For $\widehat{g}=a_n^{(1)}$, it can be understood as a $q-$deformation of the $sl_{n+1}$-Onsager algebra introduced by Uglov and Ivanov \cite{Uglov}. By analogy with the $\widehat{sl_2}$ case \cite{B1,IT32}, an algebra homomorphism from ${\cal O}_q({\widehat{g}})$ to a coideal subalgebra of the Drinfeld-Jimbo \cite{Dr,Jim} quantum universal enveloping algebra ${\cal U}_q(\widehat{g})$  is known \cite{BB1} (see also \cite{Kolb}). Realizations in terms of finite dimensional quantum algebras can be also considered (see for instance \cite{BF}): using either the coideal subalgebras of ${\cal U}_q(g)$ studied by Letzter \cite{Lez} or the non-standard ${\cal U}'_q(so_n)$ introduced by Klimyk, Gavrilik and Iorgov \cite{GI,Klim}.  Part of the motivation for the present letter comes from the fact that generalized $q-$Onsager algebras already find applications in the study of quantum integrable systems with boundaries\footnote{Integrable {\it scalar} \cite{GZ,dVG,Cor} or {\it dynamical} \cite{BasK3,BF} boundary conditions of  boundary affine Toda field theories are classified according to the representation theory of ${\cal O}_q({\widehat{g}})$ \cite{BB1}. Also, several known solutions of the reflection equations are intertwiners of ${\cal O}_q({\widehat{g}})$. For explicit examples, see e.g. \cite{DelG,DM,BF}.}, and so deserve further investigation.\vspace{1mm}

Besides the definition of the generalized $q-$Onsager algebra in terms of generators and relations \cite[Definition 2.1]{BB1}, most of its properties remain to be studied. In view of its relation with coideal subalgebras of ${\cal U}_q(\widehat{g})$ \cite{BB1,Kolb} and its application to the theory of quantum integrable systems, an important problem is to identify those properties of ${\cal U}_q(\widehat{g})$ which could be somehow extended to ${\cal O}_q({\widehat{g}})$. For instance, define the extended Cartan matrix\footnote{For  $\widehat{g}=\widehat{sl_2}$, recall that $a_{ii}=2$, $a_{ij}=-2$. For the family of simply-laced affine Lie algebras, $a_{ii}=2$, $a_{ij}=-1$ for $i,j$ simply linked and $a_{ij}=0$ otherwise \cite{Kac}. For non-simply laced cases, fix coprime integers $d_i$ such that $d_ia_{ij}$ is symmetric. We define $q_i=q^{d_i}$.} $\{a_{ij}\}$ of $\widehat{g}$. The quantum universal enveloping algebra ${\cal U}_q(\widehat{g})$ \cite{Dr,Jim}  is generated by the elements $\{h_j,e_j,f_j\}$, $j=0,1,...,rank(g)$. As shown by Lusztig \cite{Luzt}, besides the fundamental defining relations, the basic generators satisfy the so-called\footnote{As usual, we denote: $
\left[ \begin{array}{c}
n \\
m 
\end{array}\right]_q
=\frac{[n]_q!}{[m]_q!\,[n-m]_q!}\ , \quad
[n]_q!=\prod_{l=1}^n[l]_q\ ,\quad
[n]_q=\frac{q^n-q^{-n}}{q-q^{-1}}, \ \ [0]_q=1 \ .$} {\it $r-th$ higher order $q-$Serre relations}:
\beqa
\sum_{k=0}^{-a_{ij}r+1} (-1)^{k}  \,  \left[ \begin{array}{c} r+1 \\  k \end{array}\right]_{q_i}  {\bold x}_i^{-a_{ij} r+1-k} {\bold x}_j^r {\bold x}_i^k&=&0\ , \qquad {\bold x}\in\{e,f\}\ , \qquad  (i\neq j)\ .\label{hqSerre}
\eeqa
For the family of generalized $q-$Onsager algebras ${\cal O}_q({\widehat{g}})$,  by analogy with (\ref{hqSerre}),  linear relations between monomials of the fundamental generators - denoted ${\textsf A}_i$ below - are expected too. For the case $\widehat{g}=\widehat{sl_2}$, this problem has been addressed\footnote{See also \cite[Problem 3.4]{IT03}.} in \cite{BV1}: $r-th$ higher order relations satisfied by the fundamental generators ${\textsf A}_0,{\textsf A}_1$ of the $q-$Onsager algebra were conjectured, with supporting evidences described in details. In the present letter, for the family of simply-laced affine Lie algebras (of the so-called ADE-type) the $r-th$ higher order relations associated with the generalized $q-$Onsager algebras ${\cal O}_q({\widehat{g}})$ are conjectured, given by (\ref{qDGADE}) with (\ref{coefffin}). They are proven for $r\leq 5$. For $r$ generic, several supporting evidences are presented. There are three main motivations for considering this problem:
\vspace{1mm}

$\bullet$ The $r-$th higher order $q-$Serre relations (\ref{hqSerre}) arise  in the construction of a basis of ${\cal U}_q(\widehat{g})$. They also appear in the discussion of the quantum Frobenius homomorphism \cite{Luzt}. For the family of generalized $q-$Onsager algebras and corresponding coideal subalgebras of ${\cal U}_q({\widehat{g}})$ \cite{Kolb}, the $r-$th higher order relations (\ref{qDGADE}) with (\ref{coefffin}) here conjectured should play a similar role. See also \cite[Problem 3.4]{IT03} for the $\widehat{sl_2}$ case. \vspace{1mm}

$\bullet$ The generalized $q-$Onsager algebras provide a representation theoretic framework for a large class of quantum integrable systems (see e.g. \cite{Uglov,BK,BK1,BB1,BF,FK}).  The $r-$th higher order relations here proposed will allow to study the hidden symmetry of boundary quantum integrable models for {\it $q$ a root of unity}, extending the analysis conducted for ${\cal U}_q({\widehat{g}})-$related models \cite{DFM,KM} . \vspace{1mm}

$\bullet$  In \cite[Section 2]{BV1}, it was shown that the polynomial structure of the $r-$th higher order relations for the $q-$Onsager algebra (namely, the case $\widehat{g}=\widehat{sl_2}$)  is fully determined from the properties of ${\textsf A}_0,{\textsf A}_1$ acting on a finite dimensional irreducible vector space (${\textsf A}_0,{\textsf A}_1$ is called a tridiagonal pair \cite{Ter03}). By analogy, suppose some of the properties of tridiagonal pairs - in particular, the spectral properties of ${\textsf A}_i$ - can be extended to higher rank affine Lie algebras. Assuming this, the two-variable polynomial generating functions of the coefficients associated with any simply-laced $\widehat{g}$ here proposed - see Definition \ref{defpolyADE} -  should be determined from the representation theory of ${\cal O}_q(\widehat{g})$ . For $\widehat{g}=a_n^{(1)}$, this hypothesis is checked explicitly at the end of the last Section.\vspace{2mm}

This letter is organized as follows. In the next Section, inspired by the analysis of \cite[Section 3]{BV1}, the $r-th$ higher order relations  satisfied by the fundamental generators of the generalized $q-$Onsager algebra ${\cal O}_q({\widehat{g}})$ are conjectured (see Conjecture 1).  They take the form (\ref{qDGADE}). Conjecture 1 is proven for $r\leq 5$, in which case the coefficients $c^{[r,p]}_k$ are explicitly obtained in terms of the deformation parameter $q$.  Then, for generic values of $r$, an inductive argument implies the existence of a set of recursive formulae - explicitly identified - for the coefficients. These recursive relations determine uniquely the coefficients  $c^{[r,p]}_k$ in terms of $c^{[r-1,p]}_k$ . Using a  computer program,  for several values of $r\geq 6$ the conjectured $r-$th higher order relations are then checked.
In Section 3, a two-variable polynomial is proposed. It is conjectured to be the generating function for the coefficients $c^{[r,p]}_k$  (see Conjecture 2), providing the closed formula (\ref{coefffin}). For $r\leq 5$, conjecture 2  is proven. For several values of $r\geq 6$, using a computer program the closed formula for $c^{[r,p]}_k$ is found to match perfectly with the solutions of the recursive relations derived in Section 2. Additional supporting evidence are also presented.
In Appendix A, some recursion relations which occur in the analysis are reported.\vspace{1mm}

 For simplicity, here we focus on the simply-laced affine Lie algebras $\widehat{g}$. Although technically more involved, non-simply laced cases can be treated along the same line. Note that the basic defining relations of the generalized $q-$Onsager algebra associated with non-simply laced cases are given in \cite{BB1}. 
\vspace{2mm}

{\bf Notations:} Here $\{ x\}$ denotes the integer part of $x$. $q$ is assumed not to be a root of unity. 

\vspace{1mm}
\section{The higher order relations for ADE-type affine Lie algebras}
Generalized $q-$Onsager algebras are extensions of the $q-$Onsager algebra to {\it higher rank} affine Lie algebras  \cite{BB1}. Inspired by the analysis of \cite{BV1}, analogues of Lusztig's higher order relations for $\cal O_q(\widehat{g})$ can be conjectured. First, we recall some basic definitions.
\begin{defn}(see \cite{BB1}){\label{defqOnsgen}} Let $\{a_{ij}\}$ be the extended Cartan matrix of the simply-laced affine Lie algebra $\widehat{g}$.
The generalized $q-$Onsager algebra $\cal O_q(\widehat{g})$ is an associative algebra with unit $1$, elements ${A}_i$ and scalars $\rho_i$. The defining relations are:
\beqa
\quad \sum_{k=0}^{2} (-1)^k \left[ \begin{array}{c} 2 \\  k \end{array}\right]_q   {\textsf A}_i^{2-k} {{\textsf A}_j} {\textsf A}_i^{k} - \rho_{i} {\textsf A}_j &=& 0\ \quad \mbox{if} \quad i,j \quad \mbox{are linked}\ ,\label{rel1g}\\
\big[{\textsf A}_i,{\textsf A}_j\big]&=&0 \quad \mbox{otherwize}\ .\label{relautre}
\eeqa
\end{defn}
\begin{rem} For $\rho_i=0$ the relations (\ref{rel1g}) reduce to the $q-$Serre relations (\ref{hqSerre}) of ${\cal U}_{q}(\widehat{g})$. 
\end{rem}
\begin{rem}
For $q=1$, the relations (\ref{rel1g}), (\ref{relautre}) coincide with the defining relations of the so-called $sl_{n+1}-$Onsager's algebra for $n> 1$ introduced by Uglov and Ivanov \cite{Uglov}. 
\end{rem}

By analogy with the $\widehat{sl_2}$ case discussed in details in \cite{BV1}, we expect the following form for the $r-$th higher order relations:

\begin{conj} Let ${\textsf A}_j$ be the fundamental generators of $\cal O_{q}(\widehat{g})$. There exist scalars $c_{k}^{[r,p]}$ in $q$ such that:
\ben
&&\sum_{p=0}^{\{\frac{r+1}{2}\}}\sum_{k=0}^{r-2p+1} (-1)^{k+p} \, \rho_i^{p} \, \,c_{k}^{[r,p]}\,  {\textsf A}_i^{r-2p+1-k} {\textsf A}_j^r {\textsf A}_i^k=0 \  \ \quad \mbox{if} \quad i,j \quad \mbox{are linked}\ .\label{qDGADE}
\een
\end{conj}

In the first part of this Section, starting from the fundamental relations \label{rel2g} of  $\cal O_{q}(\widehat{g})$, examples of higher order relations of the form  (\ref{qDGADE}) are derived, thus proving the conjecture for $r\leq 5$. The coefficients $c_{k}^{[r,p]}$ are explicitly obtained as rational functions of the deformation parameter $q$. In a second part, the general structure of the $r-th$ higher order relations is studied. Using an inductive argument, it is shown that the coefficients 
$c_{k}^{[r,p]}$ are uniquely determined recursively in terms of $c_{k}^{[r-1,p]}$. Using these relations, for several values of $r\geq 6$ the conjecture is checked using a computer program. Note that the analysis here presented is similar to \cite[Section 3]{BV1} for the $\widehat{sl_2}$ case.\vspace{1mm}

\subsection{Proof of the $r-$th higher order relations for $r\leq 5$} For $r=1$, the relations (\ref{qDGADE}) are the defining relations of the generalized $q-$Onsager algebra $\cal O_{q}(\widehat{g})$.  Assume ${\textsf A}_i$ are the fundamental generators of $\cal O_{q}(\widehat{g})$. 
To derive the simplest example of higher order relations, we are looking for a linear relation between monomials of the type 
\beqa
{\textsf A}_i^n {\textsf A}_j^2 {\textsf A}_i^m \qquad \mbox{with}\qquad n+m=3,1 \ . 
\eeqa
Suppose it is of the form (\ref{qDGADE}) for $r=2$ with yet unknown coefficients $c_{k}^{[r,p]}$. We now show $c_{k}^{[r,p]}$ are uniquely determined in terms of $q$. First, according to the defining relations (\ref{rel1g}) the monomial ${\textsf A}_i^2{\textsf A}_j$ can be ordered as:
\beqa
{\textsf A}_i^2{\textsf A}_j = [2]_q {\textsf A}_i{\textsf A}_j{\textsf A}_i -  {\textsf A}_j{\textsf A}_i^2 + \rho_i {\textsf A}_j\  .\label{mon1}
\eeqa
Multiplying from the left and/or right by ${\textsf A}_i,{\textsf A}_j$ ($i\neq j$), the corresponding monomials can be ordered as follows: each time a monomial of the form ${\textsf A}_i^n {\textsf A}_j^2{\textsf A}_i^m$ with $n\geq 2$ arise, it is reduced using (\ref{mon1}). For instance, one has:
\beqa
{\textsf A}_i^3 {\textsf A}_j = ([2]^2_q-1) {\textsf A}_i{\textsf A}_j{\textsf A}_i^2  - [2]_q {\textsf A}_jA^3_i + \rho_i ([2]_q {\textsf A}_j{\textsf A}_i + {\textsf A}_i{\textsf A}_j) \ .
\eeqa
Now, observe that the two monomials in  (\ref{qDGADE}) for $r=2$ and $k=0,1,p=0$ can be written as ${\textsf A}_i^3 {\textsf A}_j^2\equiv ({\textsf A}_i^3 {\textsf A}_j) {\textsf A}_j$ and ${\textsf A}_i^2 {\textsf A}_j^2{\textsf A}_i \equiv({\textsf A}_i^2 {\textsf A}_j) {\textsf A}_j {\textsf A}_i$. Following the ordering prescription, each of these monomials can be reduced as a combination of monomials of the type:
\beqa
&&\quad {\textsf A}_i^n{{\textsf A}_j^2}{\textsf A}_i^m \ \qquad \ \mbox{with} \quad  \ n\leq 1\ , \ n+m=3,1\ ,\label{mongen}\\
 &&\quad {\textsf A}_i^p {\textsf A}_j {\textsf A}_i^s {\textsf A}_j {\textsf A}_i^t \quad \  \mbox{with} \quad  \ p\leq 1\ , \ s\geq 1\ ,\ p+s+t=3,1 \ .\label{badterm}
\eeqa
Plugging the reduced expressions of ${\textsf A}_i^3 {\textsf A}_j^2$ and ${\textsf A}_i^2 {\textsf A}_j^2{\textsf A}_i$  in (\ref{qDGADE}) for $r=2$, one finds that all monomials of the form (\ref{badterm}) cancel provided a simple  system of equations for the coefficients $c_{k}^{[r,p]}$ is satisfied. One finds that the solution of this system is unique, given by:
\beqa
&& c^{[2,0]}_k= \left[ \begin{array}{c} 3 \\  k \end{array}\right]_q \ \quad \mbox{for}\quad k=0,1,2,3 \ , \quad \mbox{and}\quad c^{[2,1]}_0=c^{[2,1]}_1=q^2+q^{-2}+2 \nonumber\ .
\eeqa

For $r=3,4,5$, we proceed similarly: the monomials entering in the relations (\ref{qDGADE}) are ordered according to the prescription described above. Given $r$, the  reduced expression of the corresponding relation (\ref{qDGADE}) holds provided the coefficients $c_{k}^{[r,p]}$  satisfy a system of equation which solution is unique. In each case, one finds:
\beqa
 c^{[r,0]}_k= \left[ \begin{array}{c} r+1 \\  k \end{array}\right]_q \ \quad \mbox{for}\quad k=0,...,r+1 \ ,\  r=3,4,5\ ,\label{qbin}
\eeqa
whereas for $p\geq 1$, the other coefficients are such that $c_{k}^{[r,p]}=c_{r-2p+1-k}^{[r,p]}$, given by:
\beqa
\mbox{\bf Case $r=3$:}
&& c^{[3,1]}_0= {q^4} + 2q^2+4+2q^{-2} +q^{-4} \ , \qquad c^{[3,1]}_1= [4]_q(q^2 + q^{-2}+3) \ ,\nonumber\\
&& c^{[3,2]}_0= ({q^2} + q^{-2} +1)^2 \ ; \ \nonumber \\
\mbox{\bf Case $r=4$:}&& c_0^{[4,1]}=(q^4+3+q^{-4})[2]_q^2\ ,\ c_1^{[4,1]}=[5]_q[3]_q[2]_q^2,\nonumber \\
&&c_0^{[4,2]}=(q^2+q^{-2})^2[2]_q^4 \ ;\nonumber \\
\mbox{\bf Case $r=5$:}\nonumber \\
c_0^{[5,1]}&=&q^8+2q^6+4q^4+6q^2+9+6q^{-2}+4q^{-4}+2q^{-6}+q^{-8},\nonumber\\
c_1^{[5,1]}&=&[6]_q[3]_q^{-1}(q^8+4q^6+8q^4+14q^2+16+14q^{-2}+8q^{-4}+4q^{-6}+q^{-8}),\nonumber \\
c_2^{[5,1]}&=&[6]_q[2]_q^{-1}[5]_q(q^4+3q^2+6+3q^{-2}+q^{-4}),\nonumber \\
c_0^{[5,2]}&=& q^{12}+4q^{10}+11q^8+20q^6+31q^4+40q^2+45+40q^{-2}+31q^{-4}+20q^{-6}+11q^{-8}+4q^{-10}+q^{-12},\nonumber\\
c_1^{[5,2]}&=&[6]_q[3]_q^{-1}(q^{10}+6q^8+17q^6+32q^4+47q^2+53+47q^{-2}+32q^{-4}+17q^{-6}+6q^{-8}+q^{-10}),\nonumber \\
c_0^{[5,3]}&=&[3]_q^2[5]_q^2\ . \nonumber
\eeqa

\subsection{Generic case $r$} 
Above examples suggest that $r-$th higher order relations of the form (\ref{qDGADE}) exist for generic values of $r$. The aim of this subsection is to provide strong supporting evidences for conjecture 1 with $r\geq 6$. The analysis essentially follows \cite[Section 3]{BV1}, so we refer the reader to this work for further details. First, assume that given $r$, the relation (\ref{qDGADE}) exists and that all coefficients   $c_{k}^{[r,p]}$ are already known in terms of $q$. The relation (\ref{qDGADE}) for $r\rightarrow r+1$ is then considered.  In this case, the combination
\beqa
f^{ADE}_r({\textsf A}_i,{\textsf A}_j)= {\textsf A}_i^{r+2}{\textsf A}_j^{r+1}  - c_{1}^{[r+1,0]} {\textsf A}_i^{r+1}{\textsf A}_j^{r+1}{\textsf A}_i
\eeqa
which corresponds to the two first terms in  (\ref{qDGADE})  is introduced. By analogy with the procedure described in \cite{BV1}, the monomials ${\textsf A}_i^{r+2}{\textsf A}_j^{r+1}$ and ${\textsf A}_i^{r+1}{\textsf A}_j^{r+1}{\textsf A}_i$  are reduced using (\ref{rel1g}) and  (\ref{qDGADE}). The  ordered expression of the first monomial follows:
\beqa
{\textsf A}_i^{r+2}{\textsf A}_j^{r+1}&=&\sum_{k=2}^{r+2}{(-1)^{k+1} M_k^{(r,0)}{\textsf A}_i^{r+2-k}{\textsf A}_j^r{\textsf A}_i^k{\textsf A}_j} \nonumber \\ &&+\sum_{p=1}^{\{\frac{r+1}{2}\}}\sum_{k=0}^{r+2-2p}{(-1)^{p+k+1}\rho_i^pM_k^{(r,p)}{\textsf A}_i^{r+2-2p-k}{\textsf A}_j^r{\textsf A}_i^k{\textsf A}_j}\ ,\nonumber
\eeqa
where the coefficients $M_k^{(r,p)}$ are expressed in terms of the $c^{[r,p]}_j$, as reported in Appendix A. Obviously, an ordered expression for the second monomial immediately follows from (\ref{qDGADE}). As a consequence, the whole combination $f^{ADE}_r({\textsf A}_i,{\textsf A}_j)$ can be further reduced. As an intermediate step, note that one uses  (\ref{rel1g}) to obtain:
\beqa
{\textsf A}_i^{2n+1}{\textsf A}_j&=&\sum_{p=0}^n{\rho_i^p(\eta_0^{(2n+1,p)}{\textsf A}_i{\textsf A}_j{\textsf A}_i^{2n-2p}+\eta_1^{(2n+1,p)}{\textsf A}_j{\textsf A}_i^{(2n-2p+1)})}\ , \nonumber \\
{\textsf A}_i^{2n+2}{\textsf A}_j&=&\sum_{p=0}^n{\rho_i^p(\eta_0^{(2n+2,p)}{\textsf A}_i{\textsf A}_j{\textsf A}_i^{2n+1-2p}+\eta_1^{(2n+2,p)}{\textsf A}_j{\textsf A}_i^{2n+2-2p})}+\rho_i^{n+1}{\textsf A}_j\ , \nonumber 
\eeqa
where the coefficients  $\eta_j^{(r,p)}$ are determined recursively, according to the relations given in Appendix A. The ordered expression of $f^{ADE}_r({\textsf A}_i,{\textsf A}_j)$ is then studied. If conjecture 1 holds for $r$ generic, then all coefficients of monomials of the type  ${\textsf A}_i^{p} {\textsf A}_j^r {\textsf A}_i^s {\textsf A}_j {\textsf A}_i^t$ (with $p+s+t=r,r-2,...,3,1$ if $r$ is odd, and $p+s+t=r,r-2,...,4,2$ if $r$ is even) must vanish.  Then, by straightforward calculation it implies that the coefficients  $c_{k}^{[r+1,p]}$ satisfy a finite system of equations which uniquely determine the coefficients  $c_{k}^{[r+1,p]}$ in terms of $c_{k'}^{[r,p']}$. Explicitly, we have the following results. Firstly, the coefficient of the monomial $A_i^{r+1}A_j^{r+1}A_i$ is found, given by: 
\[c^{[r+1,0]}_1=\left[ \begin{array}{c} r+1 \\ 1 \end{array}\right]_q\ .\]
According to the parity of $r$ and replacing $r+1\rightarrow r$, for the other coefficients one finds\footnote{Let $j,m,n$ be integers, we write $j=\overline{m,n}$  for $j=m,m+1,...,n-1,n$.}:\vspace{4mm}

{\bf \underline{Case $r$ odd:}} For $r=2t+1$ and $p=0$:
\beqa
c^{[2t+1,0]}_2&=& M^{(2t,0)}_2\eta^{(2,0)}_1,\nonumber \\
c^{[2t+1,0]}_{2h}&=& M^{(2t,0)}_{2h}\eta^{(2h,0)}_1+c^{[2t+1,0]}_1c^{[2t,0]}_{2h-1}\eta^{(2h-1,0)}_1,\qquad h=\overline{2,t+1},\nonumber \\
c^{[2t+1,0]}_{2h+1}&=& M^{(2t,0)}_{2h+1}\eta^{(2h+1,0)}_1+c^{[2t+1,0]}_1c^{[2t,0]}_{2h}\eta^{(2h,0)}_1, \qquad h=\overline{1,t}.\nonumber
\eeqa
Using the recursion relations given in Appendix A, it is possible to show that these coefficients can be simply written in terms of $q-$binomials: 
\beqa
 c_{k}^{[r,0]} =  \left[ \begin{array}{c} r+1 \\ k \end{array}\right]_q\ .    \label{formbinom}
\eeqa
Other coefficients $c^{[2t+1,p]}_0$ for $p \geq 1$ are determined by the following recursion relations: 
\beqa
c^{[2t+1,t+1]}_0&=&\sum_{p=0}^t{(-1)^{p+t+1}M^{(2t,p)}_{2(t+1-p)}},\nonumber \\
c^{[2t+1,1]}_0&=& -M_2^{(2t,0)}+M_0^{(2t,1)},\nonumber \\
c^{[2t+1,h]}_0&=& \sum_{p=0}^{h}{(-1)^{p+h}M^{(2t,p)}_{2(h-p)}}, \qquad h=\overline{2,t},\nonumber 
\eeqa
\beqa
c^{[2t+1,1]}_1&=&-(M^{(2t,0)}_3\eta^{(3,1)}_1-c^{[2t+1,0]}_1(-c_2^{[2t,0]}+c_0^{[2t,1]})),\nonumber \\
c^{[2t+1,h]}_1&=&\sum_{p=0}^{h-1}{(-1)^{p+h}M^{(2t,p)}_{2(h-p)+1}\eta^{(2(h-p)+1,h-p)}_1} \nonumber \\& &+c^{[2t+1,0]}_1\sum_{p=0}^{h}{(-1)^{p+h}c^{[2t,p]}_{2(h-p)}}, \qquad h=\overline{2,t},\nonumber
\eeqa
\beqa
c^{[2t+1,1]}_2&=& -M_4^{(2t,0)}\eta^{(4,1)}_1+M^{(2t,1)}_2\eta^{(2,0)}_1-c^{[2t+1,0]}_1c_3^{[2t,0]}\eta^{(3,1)}_1,\nonumber \\
c^{[2t+1,l]}_{2h-2l+1}&=&\sum_{p=0}^l{(-1)^{p+l}M^{(2t,p)}_{2(h-p)+1}}\eta_1^{(2(h-p)+1,l-p)} \nonumber\\& &+c^{[2t+1,0]}_1\sum_{p=0}^l{(-1)^{p+l}c^{[2t,p]}_{2(h-p)}\eta_1^{(2(h-p),l-p)}}, \qquad h=\overline{2,t}, \quad l=\overline{1,h-1},\nonumber \\
c^{[2t+1,l]}_{2h-2l}&=&\sum_{p=0}^l{(-1)^{p+l}M_{2(h-p)}^{(2t,p)}\eta_1^{(2(h-p),l-p)}}\nonumber\\& &+c^{[2t+1,0]}_1\sum_{p=0}^{min(l,h-2)}{(-1)^{p+l}c^{[2t,p]}_{2(h-p)-1}\eta_1^{(2(h-p)-1,l-p)}}, \qquad h=\overline{3,t+1}, \quad l=\overline{1,h-1}.\nonumber
\eeqa

{\bf \underline{Case $r$ even:}}  For $r=2t+2$, the coefficients $c^{[2t+2,0]}_j$ are given by:
\beqa
c^{[2t+2,0]}_2&=& M_2^{(2t+1,0)}\eta_1^{(2,0)},\nonumber \\
c^{[2t+2,0]}_{2h+1}&=& M_{2h+1}^{(2t+1,0)}\eta^{(2h+1,0)}_1+c^{[2t+2,0]}_1c_{2h}^{[2t+1,0]}\eta_1^{(2h,0)}, \qquad h=\overline{1,t+1},\nonumber \\
c^{[2t+2,0]}_{2h}&=& M_{2h}^{(2t+1,0)}\eta^{(2h,0)}_1+c^{[2t+2,0]}_1c_{2h-1}^{[2t+1,0]}\eta_1^{(2h-1,0)}, \qquad h=\overline{2,t+1}.\nonumber
\eeqa
According to the relations given in Appendix A,  one shows that $c^{[r,0]}_k$ simplify to $q-$binomials (\ref{formbinom}). For $p \geq 1 $, the recursive formulae for all other coefficients are given by:
\beqa
c_0^{[2t+2,1]}&=&-M^{(2t+1,0)}_2+ M_0^{(2t+1,1)},\nonumber \\
c_0^{[2t+2,h]}&=&\sum_{p=0}^{h}{(-1)^{p+h}M_{2(h-p)}^{(2t+1,p)}}, \qquad h=\overline{2,t+1},\nonumber \\
c_1^{[2t+2,1]}&=&-M_3^{(2t+1,0)}\eta_1^{(3,1)}+c_1^{[2t+2,0]}(-c^{[2t+1,0]}_2+c_0^{[2t+1,1]}),\nonumber 
\eeqa
\beqa
c_1^{[2t+2,h]}&=&\sum_{p=0}^{h-1}{(-1)^{p+h}M_{2(h-p)+1}^{(2t+1,p)}\eta_1^{(2(h-p)+1,h-p)}} \nonumber\\ & &+c_1^{[2t+2,0]}\sum_{p=0}^{h}{(-1)^{p+h}c_{2(h-p)}^{[2t+1,p]}}, \qquad h=\overline{2,t+1},\nonumber \\
c_2^{[2t+2,1]}&=&-M_4^{(2t+1,0)}\eta_1^{(4,1)}+M_2^{(2t+1,1)}\eta^{(2,0)}_1-c_1^{[2t+2,0]}c_3^{[2t+1,0]}\eta_1^{(3,1)},\nonumber
\eeqa
\beqa
c_{2h-2l}^{[2t+2,l]}&=&\sum_{p=0}^l{(-1)^{p+l}M_{2(h-p)}^{(2t+1,p)}\eta_1^{(2(h-p),l-p)}}\nonumber \\& & +c_1^{[2t+2,0]}\sum_{p=0}^{min(l,h-2)}{(-1)^{p+l}c_{2(h-p)-1}^{[2t+1,p]}\eta_1^{(2(h-p)-1,l-p)}}, \qquad h=\overline{3,t+1}, l=\overline{1,h-1},\nonumber \\
c_{2h-2l+1}^{[2t+2,l]}&=&\sum_{p=0}^l{(-1)^{p+l}M_{2(h-p)+1}^{(2t+1,p)}\eta_1^{(2(h-p)+1,l-p)}}) \nonumber \\&&+c_1^{[2t+2,0]}\sum_{p=0}^l{(-1)^{p+l}c_{2(h-p)}^{[2t+1,p]}\eta_1^{(2(h-p),l-p)}}, \qquad h=\overline{2,t+1}, l=\overline{1,h-1}.\nonumber
\eeqa
For practical purpose, for any positive integer $r$ all coefficients $c^{[r,p]}_k$ entering in the $r-th$ higher order relations (\ref{qDGADE}) can be computed recursively. Using a computer program, we have checked that the relation (\ref{qDGADE})  holds for a large number of values $r\geq 6$ provided the coefficients satisfy above recursive formulae. In addition, let us point out that that setting $\rho_i=0$, the relations (\ref{qDGADE}) reproduce the higher order $q-$Serre relations (\ref{hqSerre})  with ${\bold x}\rightarrow {\textsf A}$. 
\vspace{1mm}

\section{A two-variable polynomial generating function}
For the $q-$Onsager algebra, in \cite[Section 2]{BV1} for finite dimensional irreducible representations it was shown that the coefficients $c^{[r,p]}_k$ entering in the $r-th$ higher order relations follow from a two-variable generating function. Here, for any simply-laced affine Lie algebras we propose a two-variable generating function for the coefficients too (see conjecture 2). Various checks are done at the end of this Section, which support the proposal.  
\begin{defn}\label{defpolyADE} Let ${r\in \mathbb Z}^+$. Let $x,y$ be commuting indeterminates and $\rho$ a scalar. To any simply-laced affine Lie algebra  $\widehat{g}$, we associate the polynomial generating function  $p^{ADE}_r(x,y)$ such that:
\beqa
p_{2t+1}(x,y)&=& \prod_{l=1}^{t+1} \left(x^2- \frac{[4l-2]_q}{[2l-1]_q} xy+y^2 - \rho[2l-1]_q^2  \right) \ ,\label{polyprod1}\\
p_{2t+2}(x,y)&=& (x-y)\prod_{l=1}^{t+1} \left(x^2-  \frac{[4l]_q}{[2l]_q} xy+y^2 - \rho[2l]_q^2  \right) \ .\label{polyprod2}
\eeqa
\end{defn}
\begin{lem}\label{lem1} 
 The polynomial $p_r^{ADE}(x,y)$ can be expanded as:
\beqa
p_r^{ADE}(x,y) = \sum_{p=0}^{\{\frac{r+1}{2}\}}\ \sum_{k=0}^{r-2p+1} (-1)^{k+p}   \rho^{p}\,c_{k}^{[r,p]}\,  x^{r-2p+1-k} y^k\  \label{polyADE}
\eeqa
where the coefficients are given by:
\beqa
\qquad &&c_{k}^{[r,p]} =  \sum_{l=0}^{ \{k/\alpha\}}\frac{(  \{\frac{r+1}{2}\} - k +\alpha l -p)!}{    (\{\frac{\alpha l}{2}\})!( \{\frac{r+1}{2}\} - k + \alpha l - p -\{\frac{\alpha l}{2}\})!}
    \sum_{{\cal P}}  [s_1]^2_{q}...[s_p]^2_{q}  \frac{[2s_{p+1}]_{q}...[2s_{p+k-\alpha l}]_{q}}{[s_{p+1}]_{q}...[s_{p+k-\alpha l}]_{q}}  \label{coefffin}
\eeqa
\beqa 
\mbox{with}\quad \ \left\{\begin{array}{cc}
\!\!\! \!\!\! \!\!\! \!\!\!  \!\!\! \!\!\! \!\!\! \!\!\!   \!\!\! \!\!\! \!\!\! \!\!\! \quad \quad \quad \quad \quad \quad \quad k= \overline{0,\{\frac{r+1}{2}\}}\ , \quad s_i\in\{   r-2\{\frac{r-1}{2}\},\dots,r-2,r\}\ ,\\
{\cal P}: \begin{array}{cc} \ \ s_1<\dots<s_p\ ;\quad \ s_{p+1}<\dots<s_{p+k - \alpha l}\ ,\\
 \{s_{1},\dots,s_{p}\} \cap \{s_{p+1},\dots,s_{p+k-\alpha l}\}=\emptyset \end{array} ,\\
\alpha = \left\{ \begin{array}{cc} 1 \quad \text{if $r$ is even,} \\ \!\!\!\!2 \quad \text{if $r$ is odd}  \end{array}\right .\
\end{array}\right.\ .\nonumber
\eeqa
\end{lem}
\begin{proof}
By induction. See \cite{BV1} for the $\widehat{sl_2}$ case.
\end{proof}

We claim that the two-variable polynomial given by (\ref{polyprod1})-(\ref{polyprod2}) is the generating function for the coefficients $c^{[r,p]}_k$  entering in the higher order relations (\ref{qDGADE}). Using Lemma (\ref{lem1}),  this leads to:
\begin{conj}
 The coefficients $c^{[r,p]}_k$ in (\ref{qDGADE}) are given by (\ref{coefffin}).
\end{conj}

Supporting evidences for conjecture 2 are now presented.

\begin{itemize}
\item For $r\leq 5$, it is an exercise to check that the coefficients $c^{[r,p]}_k$ given by (\ref{coefffin}) coincide exactly with the ones derived in the previous Section (see cases $r=2,3,4,5$). This proves that conjecture 2 holds for $r\leq 5$.

\item  For $r$ generic  it is easy to check that the coefficients $c^{[r,0]}_k$ obtained from (\ref{coefffin}) are the $q-$binomials (\ref{formbinom}).

\item  For $r\geq 6$ and $p\geq 1$, the comparison is more involved. However, for a large number of values $r\geq 6$ and using a computer program we have checked that the coefficients derived using the recursive formulae coincide exactly with the ones given by (\ref{coefffin}). Although a proof of conjecture 2 for $r$ generic remains to be given,  this gives a strong support for conjecture 2.

\item  Let $\{c_i,\overline{c}_i,w_i\}\in {\mathbb C}$. Let ${\widehat{g}}=a_n^{(1)} (n>1), d_n^{(1)}, e_6^{(1)},e_7^{(1)},e_8^{(1)}$. There exists an algebra homomorphism: ${\cal O}_q(\widehat{g})\ \rightarrow \ {\cal U}_q(\widehat{g})$  given by  \cite{BB1}:
\beqa
{\textsf A}_i \longmapsto {\cal A}_i = c_i\,e_iq_i^{\frac{h_i}{2}} +\overline{c}_i\,f_iq_i^{\frac{h_i}{2}} + w_i q_i^{h_i}\ \label{realg}\
\eeqa
iff the parameters $w_i$ are subject to the constraints: 
$w_i\,\Big(w_j^2+\frac{c_j\,\overline{c}_j}{q+q^{-1}-2}\Big)=0 \ ,\ w_j\,\Big(w_i^2+\frac{c_i\,\overline{c}_i}{q+q^{-1}-2}\Big)=0$ where \ $i,j$ \ are simply linked and $\rho_i\rightarrow c_i\,\overline{c}_i $\ . Let $V$ be the so-called evaluation representation  of ${\cal U}_q(\widehat{g})$ on which  ${\cal A}_i$ act (see e.g. \cite[Proposition 1]{Jim2} for $\widehat{g}=a_n^{(1)}$). For generic parameters $c_i,\overline{c}_i,q$, $V$ is irreducible and ${\cal A}_i$ is diagonalizable on $V$. Let $\theta_k^{(i)}$, $k=0,1,...$ denote the (possibly degenerate)  corresponding eigenvalues of ${\cal A}_i$. For instance, for the fundamental representation\footnote{For $\widehat{g}=a_n^{(1)},\ n\geq 2$, see e.g. \cite{DM}. For $\widehat{g}=d_n^{(1)}, \ n\geq 4$, see e.g. \cite{DelG}.} of ${\cal U}_q(a_n^{(1)})$, the eigenvalues take the simple form:
 \beqa
\theta_k^{(i)} =  C^{(i)}(vq^k +  v^{-1}q^{-k})\ ,\label{eigenval}
\eeqa
where $v,C^{(i)}$ are scalars depending on $c_i,\overline{c}_i,q$.
Let $E_k^{(i)}$ be the projector on the eigenspace associated with the eigenvalue $\theta_k^{(i)}$. Denote $\Delta_1^{(i)}$ as the l.h.s of
the first equation in (\ref{rel1g}). The relation (\ref{rel1g}) implies that for any integers  $k,l$:
\beqa
E_k^{(i)} \Delta^{(i)}_1 E_l^{(i)} =0 \ \ \Rightarrow \ \  p^{ADE}_1(\theta_k^{(i)},\theta_l^{(i)})  E_k^{(i)} {\cal A}_j E_l^{(i)}  = 0 \quad \mbox{with} \quad \rho\equiv\rho_i \ .\nonumber
\eeqa
For generic parameters $c_i,\overline{c}_i,q$,  one observes that $E_k^{(i)} {\cal A}_j E_l^{(i)}\neq 0$ if $|k-l|\leq 1$ and is vanishing otherwize. It implies $p^{ADE}_1(\theta_k^{(i)},\theta_l^{(i)})=0$  for $l=k\pm 1$ which, according to (\ref{eigenval}), is consistent with the structure (\ref{polyprod1}) for $t=0$. The same observation about the structure of the two-variable polynomial can be generalized as follows. Denote $\Delta_r^{(i)}$ as the l.h.s of (\ref{qDGADE}). If the relation (\ref{qDGADE}) with (\ref{coefffin}) holds, then for any integers $k,l$:
\beqa
E_k^{(i)} \Delta^{(i)}_r E_l^{(i)} =0 \ \ \Rightarrow \ \   p^{ADE}_r(\theta_k^{(i)},\theta_l^{(i)})  E_k^{(i)} {\cal A}^r_j E_l^{(i)}  = 0 \quad \mbox{with} \quad \rho\equiv\rho_i  \ . \nonumber
\eeqa
For generic $c_i,\overline{c}_i,q$, $E_k^{(i)} {\cal A}_j^r E_l^{(i)}\neq 0$  if $|k-l|\leq r$ and is vanishing otherwize. It implies $p^{ADE}_r(\theta_k^{(i)},\theta_l^{(i)})=0$  if $|k-l|\leq r$. According to (\ref{eigenval}), it  leads to  the following constraints on the integers $k,l$:
\beqa
&&k=l\pm 1\ ,\ l\pm 3\ , \ l\pm 5\ ,\cdots \ ,\  l\pm r \qquad \mbox{for} \qquad r \ \ \mbox{odd} \ ,\nonumber\\
&&k=l\ ,\  l\pm 2\ , \ l\pm 4\ ,\cdots \ , \ l\pm r \qquad \quad \ \ \mbox{for} \qquad r \ \  \mbox{even} \ .\nonumber
\eeqa
Again, this is in perfect agreement with the factorized form (\ref{polyprod1}), (\ref{polyprod2}). Thus,  for $\widehat{g}=a_n^{(1)}$ the structure of the two-variable polynomial (\ref{defpolyADE}) is consistent with the spectral properties of ${\cal A}_i$.  Although the representation theory of ${\cal O}_q(\widehat{g})$ remains to be fully developed, for the simplest representations of ${\cal O}_q(a_n^{(1)})$ associated with (\ref{realg}) above analysis shows that conjectures 1 and 2 hold for $r$ generic.

\item Consider the algebra homomorphism (\ref{realg}). For $c_i=0$ or  $\overline{c}_i=0$, one has $\rho_i=0$. In this special case, the higher order relations associated with the coideal subalgebra generated by (\ref{realg}) reduce to the Lusztig's higher order $q-$Serre relations (\ref{hqSerre}) \cite{Luzt}. The coefficients $c^{[r,0]}_k$ are $q-$binomials, in agreement with \cite{Luzt}. 
\end{itemize}
\vspace{3mm}

\noindent{\bf Acknowledgements:}  P.B thanks S. Baseilhac, S. Kolb and P. Terwilliger for stimulating discussions. 

\vspace{9mm}

\centerline{\bf APPENDIX A: Coefficients $ \eta^{(m)}_{k,j}$, $M^{(r,p)}_j$}
\vspace{3mm}
The initial values of $\eta^{(m)}_{k,j}$ are given by:\\
\[\eta^{(2)}_{0,0}=[2]_q, \qquad \eta^{(2)}_{0,1}=-1.\]
The recursion relations for $\eta^{(m)}_{k,j}$  and $M^{(r,p)}_j$ read:\\
\beqa
\eta^{(2n+1)}_{p,0}&=&[2]_q\eta^{(2n)}_{p,0}+\eta^{(2n)}_{p,1}, \qquad p=\overline{0,n-1},\nonumber \\
\eta^{(2n+1)}_{n,0}&=&1,\nonumber\\
\eta^{(2n+1)}_{p,1} &=&-\eta^{(2n)}_{p,0}+\eta^{(2n)}_{p-1,0},\qquad p=\overline{1,n-1},\nonumber \\
\eta^{(2n+1)}_{0,1}&=&-\eta^{(2n)}_{0,0},\nonumber\\
\eta^{(2n+1)}_{n,1}&=&\eta^{(2n)}_{n-1,0},\nonumber \\
\eta^{(2n+2)}_{p,0}&=&[2]_q\eta^{(2n+1)}_{p,0}+\eta^{(2n+1)}_{p,1}, \qquad p=\overline{0,n},\nonumber 
\eeqa
\beqa
\eta^{(2n+2)}_{p,1}&=&\eta^{(2n+1)}_{p-1,0}-\eta^{(2n+1)}_{p,0}, \qquad p=\overline{1,n},\nonumber \\
\eta^{(2n+2)}_{0,1}&=&-\eta^{(2n+1)}_{0,0}.\nonumber
\eeqa
\beqa
M^{(2t,p)}_0&=&c^{[2t,p]}_0,\nonumber \\
M^{(2t,p)}_{2t+2-2p}&=&-c^{[2t,0]}_1c^{[2t,p]}_{2t+1-2p},\nonumber \\
M^{(2t,p)}_k&=&c^{[2t,p]}_k-c^{[2t,0]}_1c^{[2t,p]}_{k-1},\qquad k=\overline{1,2t+1-2p},\nonumber \\
M^{(2t+1,t+1)}_0&=&c^{[2t+1,t+1]}_0,\nonumber\\
M^{(2t+1,t+1)}_1&=&-c^{[2t+1,t+1]}_0c^{[2t+1,0]}_1,\nonumber \\
M^{(2t+1,p)}_k&=&-c^{[2t+1,0]}_1c^{[2t+1,p]}_{k-1}+c^{[2t+1,p]}_k,\qquad k=\overline{1,2t+2-2p},\nonumber \\
M^{(2t+1,p)}_0&=&c^{[2t+1,p]}_0,\nonumber\\
M^{(2t+1,p)}_{2t+3-2p}&=&-c^{[2t+1,0]}_1c^{[2t+1,p]}_{2t+2-2p}.\nonumber
\eeqa

\end{document}